\documentclass{llncs}

\usepackage[ruled,vlined]{algorithm2e}
\usepackage{amsmath,amssymb}
\usepackage[ruled,vlined]{algorithm2e}
\usepackage{pgf, tikz}
\usetikzlibrary{arrows}
\usepackage{xspace, units}
\allowdisplaybreaks

\usepackage{nicefrac}

\usepackage{color}
\definecolor{Darkblue}{rgb}{0,0,0.4}
\definecolor{Brown}{cmyk}{0,0.81,1.,0.60}
\definecolor{Purple}{cmyk}{0.45,0.86,0,0}
\usepackage[breaklinks]{hyperref}
\hypersetup{colorlinks=true,
            citebordercolor={.6 .6 .6},linkbordercolor={.6 .6 .6},
            citecolor=black,urlcolor=Darkblue,linkcolor=black,
            pdfpagemode=UseThumbs,pdfstartview=FitH}
\newcommand{\lref}[2][]{\hyperref[#2]{#1~\ref*{#2}}}

\newcommand{\RR}{\ensuremath{\mathbb{R}}}

\newcommand{\calX}{\mathcal{X}}

\newcommand{\e}{{\mathrm e}}

\renewcommand{\Pr}[1]{\mbox{\rm\bf Pr}\left[#1\right]}
\newcommand{\Ex}[1]{\mbox{\rm\bf E}\left[#1\right]}
\newcommand{\E}{\mbox{\rm\bf E}}

\newcommand{\replaced}[2]{
 \textcolor{black}{#2}
}

\usepackage{todonotes}
\usepackage[normalem]{ulem}

\newcommand{\growingmid}{\mathrel{}\middle|\mathrel{}}

\newenvironment{myproof}{\paragraph{\textit{\textmd{Proof of Theorem 1.}}}}

\begin{document}
\title{Smoothness for Simultaneous Composition of\\Mechanisms with Admission\thanks{This work is supported by DFG through Cluster of Excellence MMCI.}}
\author{Martin Hoefer \and Thomas Kesselheim \and Bojana Kodric}
\institute{MPI Informatik and Saarland University, Saarbr\"ucken, Germany\\
\email{\{martin.hoefer,thomas.kesselheim,bojana.kodric\}@mpi-inf.mpg.de}}

\maketitle

\begin{abstract}
We study social welfare of learning outcomes in mechanisms with admission. In our repeated game there are $n$ bidders and $m$ mechanisms, and in each round each mechanism is available for each bidder only with a certain probability. Our scenario is an elementary case of simple mechanism design with incomplete information, where availabilities are bidder types. It captures natural applications in online markets with limited supply and can be used to model access of unreliable channels in wireless networks.
If mechanisms satisfy a smoothness guarantee, existing results show that learning outcomes recover a significant fraction of the optimal social welfare. These approaches, however, have serious drawbacks in terms of plausibility and computational complexity. Also, the guarantees apply only when availabilities are stochastically independent among bidders.
In contrast, we propose an alternative approach where each bidder uses a single no-regret learning algorithm and applies it in all rounds. This results in what we call \emph{availability-oblivious} coarse correlated equilibria. It exponentially decreases the learning burden, simplifies implementation (e.g., as a method for channel access in wireless devices), and thereby addresses some of the concerns about Bayes-Nash equilibria and learning outcomes in Bayesian settings. Our main results are general composition theorems for smooth mechanisms when valuation functions of bidders are lattice-submodular. They rely on an interesting connection to the notion of correlation gap of submodular functions over product lattices.
\end{abstract}
%\keywords{Smoothness $\cdot$ Price of Anarchy $\cdot$ Bayesian Coarse Correlated Equilibrium $\cdot$ Correlation Gap}
% !TEX root = main.tex

\section{Introduction}

Truthful mechanism design is a central challenge at the intersection of economics and computer science, but many fundamental techniques are only very rarely used in practice. For example, sponsored search auctions are used on a daily basis and generate billions of dollars in revenue, but they are based on simple and non-truthful procedures to allocate ads on search result pages. In contrast, truthful mechanisms often involve heavy algorithmic machinery, complicated allocation techniques, or other hurdles to easy and transparent implementation. %To a different extent, the same is true for truthful mechanisms in combinatorial auctions that provide good approximations for revenue or social welfare. These mechanisms rely on heavy algorithmic machinery (solving convex programs or decompositions based on the ellipsoid method) and are not easily applicable in a time-sensitive or dynamic environment.

A recent trend is to study non-truthful and conceptually ``simple'' mechanisms for allocation in markets and their inherent loss in system performance. The idea is to analyze the induced game among the bidders and bound the quality of (possibly manipulated) outcomes in equilibrium. In a seminal paper, Syrgkanis and Tardos \cite{SyrgkanisT13} propose a general technique for bounding social welfare of these equilibria, based on a so-called ``smoothness'' technique. These guarantees apply even to mixed Bayes-Nash equilibria in environments with composition of mechanisms. For example, in a combinatorial auction we might not sell all items via a complicated truthful mechanism, but instead sell each item simultaneously via simple individual single-item auctions. Such a mechanism is obviously not truthful, since bidders are not even able to express their valuations for all subsets of items. However, if bidders have complement-free XOS valuations, the (expected) social welfare of allocations in a mixed Bayes-Nash equilibrium turns out to be a constant-factor approximation of the optimal social welfare. 

While this is a fundamental insight into non-truthful mechanisms, it is not well-understood how this result extends under more realistic conditions. In particular, there has been recent concern about the plausibility and computational complexity of exact and approximate Bayes-Nash equilibria~\cite{CaiP14}. For more general Bayesian concepts based on no-regret learning strategies in repeated games, there are two natural approaches -- either bidder types are drawn newly with bids, or types are drawn only once initially. While the latter is not really in line with the idea of incomplete information (bidders could communicate their type in the course of learning, see~\cite{CaiP14}), the former is in general hard to obtain. Also, the composition theorem applies only if bidders' types are drawn independently.

In this paper, we study a variant of simultaneous composition of mechanisms and show how to avoid the drawbacks of the Bayesian approach. Our scenario is motivated by limited availability or admission: Suppose bidders try to acquire items in a repeated online market, in which $m$ items are sold simultaneously via, say, first-price auctions. However, in each round only some of the items are actually available for purchase. This scenario can be phrased in the Bayesian framework when bidder $i$'s type is given by the set of items available to him. To obtain an equilibrium in the Bayesian sense, each bidder would have to consider a complicated bid vector and satisfy an equilibrium condition for each of the possible $2^m$ subsets of items. %Also, to apply the composition theorem above, the availabilities of items to bidders are drawn independently.

In contrast, here we assume that bidders do not even get to know (or are not able to account for) \emph{their own availabilities} before making bids in each round. We assume they learn with no-regret strategies in a way that is \emph{oblivious} to their own and all other bidders' availabilities. Thereby, bidders arrive at what one might term an \emph{availability-oblivious} coarse-correlated equilibrium -- a bid distribution not tailored to the specific availabilities of bidders, which can be computed (approximately) in polynomial time. Our main result is that for a large class of valuation functions, we can apply smoothness ideas in this framework and prove bounds that mirror the guarantees above. The guarantees apply even if some bidders learn obliviously and others follow a Bayes-Nash bidding strategy. In particular, we cover a broad domain with simultaneous composition of weakly smooth mechanisms in the sense of~\cite{SyrgkanisT13} when bidders have lattice-submodular valuations. Our study covers cases where availabilities are correlated among bidders and provides lower bounds for combinatorial auctions with item-bidding and XOS valuations. As a part of our analysis, we use the concept of correlation gap from~\cite{AgrawalDis11} for submodular functions over product lattices. %Our proof simplifies and slightly extends the result for submodular functions in~\cite{AgrawalDis11}.

%Perhaps surprisingly, 
%Our results also have implications beyons auctions, e.g., for spectrum access in wireless networks. In recent years, no-regret learning algorithms have been successfully applied to coordinate secondary spectrum access~\cite{AsgeirssonM11,DamsHK13}. Here each bidder is a wireless device, each item corresponds to a channel. We show how to interpret this approach as a composition of smooth all-pay auctions, where each channel is sold independently via a separate all-pay auction to a conflict-free subset of devices. Availability becomes a very natural condition: A channel is unavailable to a device if it is currently congested by a different subsystem or blocked due to exclusive access of a licensed primary user. Even though it ignores channel availability, our simple oblivious learning approach can obtain a competitive solution with significant throughput.

\subsection{Our Contribution}

We assume that every mechanism satisfies a weak smoothness bound (for more details see Section~\ref{sec:model} below) with parameters $\lambda, \mu_1, \mu_2\ge0$. It is known that for each individual mechanism, this implies an upper bound of $(\max(1,\mu_1)+\mu_2)/\lambda$ on the price of anarchy for no-regret learning outcomes and Bayes-Nash equilibria. Furthermore, the same bound also applies for outcomes of multiple simultaneous mechanisms that are tailored to availabilities, i.e., not oblivious.
%The price of anarchy for no-regret learning outcomes tailored to availabilities becomes $(\max(1,\mu_1)+\mu_2)/\lambda$, and this extends to composition of mechanisms and Bayes-Nash equilibria.

In Section~\ref{sec:independent} we consider smoothness for oblivious learning and composition with independent availabilities, where in each round $t$, each mechanism $j$ is available to each bidder $i$ independently with probability $q_{i,j}$. Our smoothness bound involves the above parameters and the correlation gap of the class of valuation functions. In particular, if valuations $v_i$ come from a class $\mathcal{V}$ with a correlation gap of $\gamma(\mathcal{V})$, the price of anarchy becomes $\gamma(\mathcal{V})\cdot(\max(1,\mu_1)+\mu_2)/\lambda$. 

Our construction uses smoothness of simultaneous composition from~\cite{SyrgkanisT13}. However, since learning is oblivious, the deviations establishing smoothness must be independent of availability. Here we use correlation gap to relate the value for independent deviations to that of type-dependent Bayesian deviations. Correlation gap is a notion originally defined for submodular set functions in~\cite{AgrawalDSY10}. It captures the worst-case ratio between the expected value of independent and correlated distributions over elements with the same marginals. We use an extension of this notion from~\cite{AgrawalDis11} to Cartesian products of outcome spaces such as product lattices. For the class $\mathcal{V}$ of monotone lattice-submodular valuations, we prove a correlation gap of $\gamma(\mathcal{V}) = e/(e-1)$, which simplifies and slightly extends previous results.

In Section~\ref{sec:correlated}, we analyze oblivious learning for composition with correlated availabilities in the form of ``everybody-or-nobody'' -- each mechanism is either available to all bidders or to no bidder. The probability for availability of mechanism $j$ is $q_j$, and availabilities are independent among mechanisms. In this case, %we cannot directly relate our result to bounds for mixed Bayes-Nash equilibria, since the types of bidders are now correlated. Instead, 
we simulate independence by assuming that each bidder draws random types and outcomes for himself. We also consider distributions where outcomes are drawn independently according to the marginals from the optimal correlated distribution over outcomes. While these two distributions are directly related via correlation gap, the technical challenge is to show that there is a connection to the value obtained by the bidder. For lattice-submodular functions, we show a smoothness bound that implies a price of anarchy of $4e/(e-1)\cdot(\max(1,\mu_1)+\mu_2)/\lambda^2$. 

For neither of the results is it necessary that all bidders follow our oblivious-learning approach. We only require that bidders have no regret compared to this strategy. This is also fulfilled if some or all bidders determine their bids based on the actually available items rather than in the oblivious way.

Finally, in Section~\ref{sec:xosLower} we show a lower bound for simultaneous composition of single-item first-price auctions with general XOS valuation functions. The correlation gap for such functions is known to be large~\cite{AgrawalDSY10}, but this does not directly imply a lower bound on the price of anarchy for oblivious learning. We provide a class of instances where the price of anarchy for oblivious learning becomes $\Omega((\log m)/(\log m \log m))$. This shows that for XOS functions it is impossible to generalize the constant price of anarchy for single-item first-price auctions. 

Our results have additional implications beyond auctions for the analysis of regret learning in wireless networks. We discuss these in Appendix~\ref{sec:wireless}.

\subsection{Related Work}

Closely related to our work are combinatorial auctions with item bidding, where multiple items are being sold in separate auctions. Bidders are generally interested in multiple items. However, depending on the bidder, some items may be substitutes for others. As the auctions work independently, bidders have to strategize in order to buy not too many items simultaneously. In a number of papers~\cite{ChristodoulouKS08,BhawalkarR11,HassidimKMN11,FeldmanFGL13} the efficiency of Nash and Bayes-Nash equilibria has been studied. It has been shown that, if the single items are sold in first or second price auctions and if the valuation functions are XOS or subadditive, the price of anarchy is constant. Limitations of this approach are shown in~\cite{ChristodoulouKST16,Roughgarden14}.

Many of these proofs follow a similar pattern, namely showing smoothness. This concept has been introduced by \cite{Roughgarden12CACM,Roughgarden12} to analyze correlated and Bayes-Nash equilibria of general games. In~\cite{SyrgkanisT13} it was adjusted to mechanisms, and it was shown that simultaneous or sequential composition of smooth mechanisms is again smooth. Combinatorial auctions with item bidding are an example of a simultaneous composition. To show smoothness of the combined mechanism, it is thus enough to show smoothness of each single auction. Other examples of smooth mechanisms are position auctions with generalized second price \cite{CaragiannisKKKLLT15,LemeT10} and greedy auctions \cite{LucierB10}. The smoothness approach for mixed Bayes-Nash equilibria shown in~\cite{SyrgkanisT13} is, in fact, slightly more general and continues to hold for variants of Bayesian correlated equilibrium~\cite{HartlineST15}.

The complexity of finding such equilibria has been studied only very recently. It has been shown in~\cite{CaiP14,DobzinskiFK15} that equilibria are hard to find in some settings. In contrast, in~\cite{DevanurMSW15} a different auction format is studied that yields good bounds on social welfare for equilibria that can be found more easily. Although similar in spirit, our approach is different -- it shows that in some scenarios agents can reduce the computational effort and still obtain reasonably good states with existing mechanisms. 

As such, our approach is closer to recent work~\cite{DaskalakisS16} that shows hardness results for learning full-information coarse-correlated equilibria in simultaneous single-item second-price auctions with unit-demand bidders. As a remedy, a form of so-called no-envy learning is proposed, in which bidders use a different form of bidding that enables convergence in polynomial time. While achieving a general no-regret guarantee against all possible bid vectors is hard, we note here that our approach based on smoothness requires only a guarantee with respect to bids that are derived directly from the XOS representation of the bidder valuation. As such, bidders can obtain the guarantees required for our results in polynomial time. Conceptually, we here treat a different problem -- the impact of availabilities, and more generally, different bidder types on learning outcomes in repeated mechanism design.

A model with dynamic populations in games has recently been considered in~\cite{LykourisST16}. Each round a small portion of players are replaced by others with different utility functions. When players use algorithms that minimize a notion called adaptive regret, smoothness conditions and the resulting bounds on the price of anarchy continue to hold if there are solutions which remain near-optimal over time with a small number of structural changes. Using tools from differential privacy, these conditions are shown for some special classes of games, including first-price auctions with unit-demand or gross-substitutes valuations. In contrast, our scenario is orthogonal, since we consider much more general classes of mechanims and allow changes in each round for possibly all players. However, our model of change captures the notion of availability and therefore is much more specific than the adversarial approach of~\cite{LykourisST16}.

The notion of correlation gap was defined and analyzed for stochastic optimization in~\cite{AgrawalDSY10,AgrawalDis11}. The notion was used in~\cite{Yan11} for analyzing revenue maximization with sequential auctions, which is very different from our approach.

\section{Model and Preliminaries}
\label{sec:model}
There are $n$ bidders that participate in $m$ simultaneous mechanisms. Each mechanism $j \in [m]$ is a pair $M_j = (f_j, p_j)$, consisting of an outcome function and payment functions. More formally, function $f_j\colon B_j \to \calX_j$ maps every bid vector $b_{\cdot,j}$ on mechanism $j$ into an outcome space $\calX_j$. The function $p_j = (p_{1, j},\ldots,p_{n, j})$ defines a payment for each bidder. That is, depending on the bid vector, $p_{i, j} \colon B_j \to \RR_{\geq 0}$ defines the non-negative payment for bidder $i$ in mechanism $j$.

We consider a repeated framework with oblivious learning in a simultaneous composition of mechanisms with availabilities. There are $T$ rounds and in each round the bidders participate in $m$ simultaneous mechanisms. In round $t=1,\ldots,T$, each bidder places a bid $b_{i,j}^t$ for each mechanism, the mechanism determines the outcome and the payments, and bidder $i$ has a utility function $u_i(b^t) = v_i(f(b^t)) - p_i(b^t)$, where $v_i$ is a valuation function over vectors of outcomes and $p_i = \sum_{j} p_{i,j}(b^t)$. In addition, in each round we assume that each mechanism is available to each bidder with a certain probability. We let the Bernoulli random variable $A_{i,j} = 1$ if mechanism $j$ is available to bidder $i$. Due to availability, the mechanisms must also be applicable when only subsets of bidders are placing bids. For this reason, it will be convenient to assume that the outcome space for mechanism $j \in [m]$ is $\calX_j = \calX_{1,j} \times \ldots \times \calX_{n,j}$ and $x_j \in \calX_j$ is $x_j = (x_{i,j})_{i \in [n]}$. We assume that each bidder, for whom the mechanism is not available, must place a bid of ``0''. If bidder $i$ bids 0 for mechanism $j$, we assume $f_j(0,b_{-i,j}) = \bot_{i,j}$, where $\bot_{i,j}$ is a ``losing'' outcome, and payment $p_{i,j}(0,b_{-i,j}) = 0$. For convenience, we will denote by $f = (f_j)_{j \in [m]}$ the composed mechanism and by $\mathcal{X} = \mathcal{X}_1 \times \ldots \times \mathcal{X}_m$ its outcome space.

\paragraph{Oblivious Learning}
We assume \emph{oblivious learning} -- each bidder runs a single no-regret learning algorithm and uses the utility of every round as feedback, no matter how the availability in each round turns out. In hindsight, the average history of play for oblivious learning becomes an availability-oblivious variant of coarse-correlated equilibrium~\cite{BlumChapter07}. Hence, the outcomes of oblivious learning are captured by the coarse-correlated equilibria in the following one-shot game: First, all bidders simultaneously place a bid for every mechanism. They know only the probability distribution of the availabilities. Only after they placed their bids, the availability of each mechanism for each bidder is determined at random. 

\begin{definition}
\label{def:obliviousCCE}
An \emph{availability-oblivious coarse-correlated equilibrium} is a distribution over bid vectors $b$ (independent of $A$) such that, in expectation over all availabilities, it is not beneficial for any bidder $i$ to switch to another bid $b_i'$. For each $i$ and each $b_i'$, we have $\E\left[u_i(b_i', b_{-i})\right] \leq \E\left[u_i(b)\right]$.
\end{definition}

Indeed, our results also hold for a larger class of equilibria, in which a subset of bidders might not be oblivious to availabilities. For our guarantees, it is enough to consider distributions over bidding strategies $b$ which might dependent on $A$ such that, in expectation over all availabilities, it is not beneficial for any bidder $i$ to switch to another bid $b_i'$. For each $i$ and each $b_i'$, we have $\E\left[u_i(b_i', b_{-i})\right] \leq \E\left[u_i(b)\right]$. Note that both ordinary coarse-correlated equilibria and availability-oblivious ones fulfill this property.

We bound the performance of these equilibria by deriving suitable smoothness bounds.

\paragraph{Smoothness}
We assume that each mechanism $j$ satisfies \emph{weak smoothness} as defined in~\cite{SyrgkanisT13}. For any valuations $v_{i, j}\colon \calX_j \to \RR^{\ge 0}$ there are (possibly randomized) deviations\footnote{In slight contrast to~\cite{SyrgkanisT13}, we here assume that the smoothness deviations of a bidder do not depend on his own current bid. This serves to simplify our exposition and can be incorporated into our analysis.} $b'_{i, j}$ for each $i \in [n]$ such that for all bid vectors $b_{\cdot, j}$
\begin{multline}
	\label{eq:smoothness}
	\E\left[\sum_{i \in [n]} v_{i, j}(f_j(b'_{i, j}, b_{-i,j})) - p_{i, j}(b'_{i, j}, b_{-i, j}) \right] \\ \ge \lambda \cdot \max_{x_j \in \calX_j} \sum_{i \in [n]} v_{i, j}(x_j) - \mu_1 \cdot \sum_{i \in [n]} p_{i,j}(b_{\cdot,j}) - \mu_2 \sum_{i \in [n]} h_{i, j}(b_{i, j}, f_j(b_{\cdot,j})) \enspace,
\end{multline}
where $h_{i, j}(b_{i, j}, x_j) = \max_{b_{-i, j}: f_j(b_{\cdot,j}) = x_j} p_{i, j}(b_{\cdot,j})$. 
For intuition, assume that \eqref{eq:smoothness} holds with $\mu_2 = 0$. Consider a learning outcome with a no-regret guarantee where every bidder $i$ can gain at most $\epsilon$ in any fixed deviation, i.e., $\E[v_{i,j}(f_j(b_{\cdot,j})) - p_{i,j}(b_{\cdot,j})] \ge \E\left[v_{i,j}(f_j(b'_{i,j},b_{-i,j})) - p_{i,j}(b'_{i,j}, b_{-i,j})\right] - \epsilon\enspace$. Applying~\eqref{eq:smoothness} pointwise
\begin{equation*}
\sum_{i \in [n]} \E[v_{i,j}(f_j(b_{\cdot,j})) - p_{i,j}(b_{\cdot,j})] \ge \lambda \cdot \max_{x_j \in \calX_j} \sum_{i \in [n]} v_{i, j}(x_j) - \mu_1 \cdot \sum_{i \in [n]} \E[p_{i,j}(b_{\cdot,j})] - n\epsilon,
\end{equation*}
which implies for social welfare
\begin{equation*}
 \sum_{i \in [n]} \E\left[v_{i,j}(f_j(b_{\cdot,j}))\right] \quad \ge \quad \lambda \cdot \max_{x_j \in \calX_j} \sum_{i \in [n]} v_{i, j}(x_j) + (1-\mu_1) \cdot \sum_{i \in [n]} \E[p_{i,j}(b_{\cdot,j})] - n\epsilon\enspace.
\end{equation*}
Every bidder $i$ can stay away from the market and payments are non-negative, so $0 \le \E[p_{i,j}(b_{\cdot,j})] \le \E[v_{i,j}(f_j(b_{\cdot,j}))]+\epsilon$ and 
\[
\max(1,\mu_1) \sum_{i \in [n]} \E\left[v_{i,j}(f_j(b_{\cdot,j}))\right] \quad \ge \quad \lambda \cdot \max_{x_j \in \calX_j} \sum_{i \in [n]} v_{i, j}(x_j) - (n+ \mu_1)\epsilon\enspace.
\]
Thus, for $\epsilon \to 0$, the price of anarchy tends to $\max(1,\mu_1)/\lambda$. More generally, \eqref{eq:smoothness} implies a bound on the price of anarchy of $(\mu_2 + \max(1,\mu_1))/\lambda$ for many equilibrium concepts. If $\mu_2 > 0$, then the bound relies on an additional no-overbidding assumption, which directly transfers to our results. For details see~\cite{SyrgkanisT13}. 

\paragraph{Valuation Functions}
Our main results apply for the class of monotone lattice-submodular valuations.
%
%\begin{definition}[\cite{SyrgkanisT13}]
Suppose for every mechanism $j$ the set $\calX_{ij}$ of possible outcomes for bidder $i$ forms a lattice $(\calX_{ij}, \succeq_{ij})$ with a partial order $\succeq_{ij}$. Bidder $i$ has a \emph{lattice-submodular valuation} $v_i$ if and only if it is submodular on the product lattice $(\calX_i, \succeq_i)$ of outcomes for bidder $i$: $\forall x_i, \tilde{x}_i \in \calX_i : v_i(x_i \vee \tilde{x}_i) + v_i(x_i \wedge \tilde{x}_i) \le v_i(x_i) + v_i(\tilde{x}_i)$. In the paper, we concentrate on distributive lattices, for which this definition is equivalent to the diminishing marginal returns property: 
\[ \forall z_i \succeq_i y_i \in \calX_i \Longrightarrow \forall t \in \calX_i : v_i(t \vee y_i) - v(y_i) \ge v_i(t \vee z_i) - v(z_i). \]
%\end{definition}
%
Lattice-submodular functions generalize submodular set functions but are a strict subclass of XOS functions. 
%
%\begin{definition}
Bidder $i$ has an \emph{XOS valuation} $v_i$ if and only if there are additive functions $v_i^1, v_i^2, \ldots$ with $v_i^{k_i}(x_i) = \sum_{j} v_{ij}^{k_i}(x_{ij})$ for every $x_{i, j} \in \calX_{i, j}$ and $v_i(x_i) = \max_{k_i} v_i^{k_i}(x_i)$. 
%\end{definition}
%
%\input{relatedwork}
% !TEX root = main.tex

\section{Composition with Independent Admission}
\label{sec:independent}

%\subsection{Smoothness for Lattice-Submodular Valuations}
We first consider simultaneous composition of smooth mechanisms with independent availabilities. Here, all random variables $A_{i,j}$ are independent, and we let $q_{i,j} = \Pr{A_{i,j} = 1}$.

\begin{definition}\label{def:correlation-gap}
Let $v$ be a valuation function on a product lattice, coming from a class of valuation functions $\mathcal{V}$. Given vectors $x^1, \ldots, x^k$ and numbers $\alpha_1, \ldots, \alpha_k \in [0,1]$ such that $\sum_{j = 1}^k \alpha_j = 1$, determine another vector $y$ at random by setting component $y_i$ to $x^j_i$ independently with probability $\alpha_j$. Then, the smallest $\gamma$ s.t. $\sum_{j=1}^k \alpha_j v(x^j) \le \gamma\cdot \E\left[v(y)\right]$ is the \emph{correlation gap} of class $\mathcal{V}$.
\end{definition}

\begin{theorem}\label{thm:independent}
Suppose bidder valuations are monotone and come from a class $\mathcal{V}$ with a correlation gap of $\gamma(\mathcal{V})$.	The price of anarchy for oblivious learning for simultaneous composition of weakly $(\lambda,\mu_1,\mu_2)$-smooth mechanisms with valuations from $\mathcal{V}$ and fully independent availability is at most $\gamma(\mathcal{V})\cdot(\mu_2 + \max(1,\mu_1))/\lambda$.
\end{theorem}

Before the proof of the main theorem of this section, we note that in Appendix~\ref{app:correlationGap} we also prove an upper bound of $e/(e-1)$ on the correlation gap of lattice-submodular valuations with diminishing marginal returns. This result slightly generalizes the result of~\cite{AgrawalDis11} from composition of totally ordered sets to arbitrary product lattices.

\begin{lemma}[Correlation Gap on a Product Lattice]
\label{lemma:gapLatticeIndep}
Let $v$ be a function with diminishing marginal returns on a product lattice. Given vectors $x^1, \ldots, x^k$ and numbers $\alpha_1, \ldots, \alpha_k \in [0,1]$ such that $\sum_{j = 1}^k \alpha_j = 1$, determine another vector $y$ at random by setting component $y_i$ to $x^j_i$ independently with probability $\alpha_j$. Then $\Ex{v(y)} \geq \left( 1 - \frac{1}{\e} \right) \sum_{j=1}^k \alpha_j v(x^j)$.
\end{lemma}

\noindent From here, we arrive at the following corrolary of the main theorem.

% Note that the proof uses lattice-submodular valuations only as a prerequisite for a bounded correlation gap. Hence, more generally, we can use the same proof for every class of valuation functions by replacing $(1-1/\e)$ with a suitable upper bound on the correlation gap.

\begin{corollary}
The price of anarchy for oblivious learning for simultaneous comp\-osition of weakly $(\lambda,\mu_1,\mu_2)$-smooth mechanisms with monotone lattice-submodular valuations and fully independent availability is at most $e/(e-1)\cdot(\mu_2 + \max(1,\mu_1))/\lambda$.
\end{corollary}

\begin{myproof}
We will prove the theorem by defining an availability-oblivious (randomized) deviation $b_i'$ for each player $i$ such that the following inequality will hold for any (not necessarily availability-oblivious) bidding strategy $b$:
\begin{align}
&\sum_i \E\left[ u_i(b_i', b_{-i}) \right] \nonumber\\
&\geq \quad \frac{1}{\gamma(\mathcal{V})} \cdot \lambda \cdot \sum_i \E\left[ v_i(x^*) \right] - \mu_1 \sum_i \E\left[ p_i(b) \right] - \mu_2 \sum_i \E\left[ h_i(b_i, f(b))\right]\enspace,
\label{eq:independent_smoothnes}
\end{align}
where $x^*$ denotes the (random) optimal outcome. From this inequality, whose form is in fact exactly that of the smoothness condition (\ref{eq:smoothness}), the claim of the theorem follows as described in Section~\ref{sec:model}.

In more detail, to attain the aforementioned inequality, we will relate each player's utility for deviating to $b_i'$ to the utility he could achieve if he was allowed to see and react upon the availabilities. In that case, he could simply use the smoothness deviation tailored to the specific availability profile $A_i = (A_{i,1},\ldots,A_{i,m})$ that he is encountering. We denote this non-oblivious smoothness deviation by $b_i^{A_i}$. Because the global mechanism is a simultaneous composition of $(\lambda,\mu_1,\mu_2)$-smooth mechanisms, it is again $(\lambda,\mu_1,\mu_2)$-smooth. Therefore we know that the non-oblivious deviations $b_i^{A_i}$ do exist, and they satisfy the smoothness inequality (\ref{eq:smoothness}) by definition. 

We proceed to define, for each player $i$, the availability-oblivious deviation $b_i'$. First, bidder $i$ assumes for himself a reduced valuation function $\bar{v}_i = \alpha \cdot v_i$, for some appropriate $\alpha$ to be chosen later. The deviation $b_i'$ is a composition of component-wise independent deviations $b_{i, j}'$, i.e. $b_i'=(b_{i, 1}', \dots, b_{i, m}')$ where each $b_{i, j}'$ is chosen independently. To arrive at $b_{i, j}'$, bidder $i$ assumes that mechanism $j$ is available to him and draws all other availabilities independently according to probabilities $q_{i', j'}$. This means that he draws availabilities for all other players on all mechanisms and also his own availabilities on all mechanisms other than $j$. Now he has a full availability profile, and therefore he can consider the non-oblivious smoothness deviation. He observes the $j$-th component of this smoothness deviation and sets $b_{i, j}'$ to be equal to it. Note that $b_{i, j}'$ will be applied only with the probability that mechanism $j$ is in fact available to bidder $i$, i.e. with probability $q_{i, j}$.

Next, we want to compare $u_i(b_i', b_{-i})$ and $u_i(b_i^{A_i}, b_{-i})$. Let us focus on the valuation $v_i(f(b_i', b_{-i}))$ first. The non-oblivious smoothness deviation $b_i^{A_i}$ is a vector whose components are correlated. More precisely, to form this bid we observe $A_i$, sample the availabilities $A_{-i}$ and bids $b_{-i}$ of other players, and take the optimal allocation $x^*$ for the resulting availability profile $A$. Then, we determine the $\ell$ for which $\bar{v}_i(x^*_i)=\sum_j \bar{v}_{i, j}^{\ell}(x^*_{i, j})$ and use $\bar{v}_{i, j}^{\ell}$ for determining $b_{i, j}^{A_i}$ (note that $A_i$ can be regarded as bidder $i$'s type in a Bayesian sense, for more details see \cite{SyrgkanisT13}). Therefore, the components of $b_i^{A_i}$ are correlated through the common choice of $\ell$. Our deviation $b_i'$ is assembled by setting $b_{i, j}' = (b_{i, j}^{A_i})_{k_j}$ independently for each $j$.

Formally, let $r_{i, j}^{\ell}$ denote the conditional probability that the optimum yields an outcome vector $x^*$ that attains its maximum value for bidder $i$ in $\bar{v}_i^{\ell}$, given that $A_{i,j} = 1$. Then, the marginal probability of observing $b_{i, j}^{A_i} = (b_{i, j}^{A_i})_{\ell}$ is $r_{i,j}^{\ell}q_{i,j}$. In $b_i'$ we pick $\ell$ independently for each mechanism with probability $r_{i,j}^{\ell}$, which yields a combined probability of $r_{i,j}^{\ell}q_{i,j}$ for availability and deviation. Thus, $b_i'$ simulates the marginal probabilities of outcomes in $b_i^{A_i}$, i.e., $\Pr{f_j(b'_i,b_{-i}) = y_{i,j} \mid A_{-i}, b_{-i}} = \Pr{f_j(b^{A_i}_i,b_{-i}) = y_{i,j} \mid A_{-i}, b_{-i}}$ for all $y_{i,j} \in \calX_{i,j}$, for each $j \in [m]$. Hence, for fixed $A_{-i}, b_{-i}$, the two expected valuations $\E\left[ v_i(f(b'_i,b_{-i})) \mid A_{-i}, b_{-i} \right]$ and $\E\left[ v_i(f(b_i^{A_i},b_{-i})) \mid A_{-i}, b_{-i} \right]$ are related via correlation gap.

\noindent Thus, setting $\alpha = 1/\gamma(\mathcal{V})$ and $\bar{v}_i(x) = 1/\gamma(\mathcal{V})\cdot v_i(x)$ we get
\begin{align*}
\E\left[ v_i(f(b'_i,b_{-i})) \mid A_{-i}, b_{-i} \right] 
& = \quad \sum_{y \in \mathcal{X}} v_i(y) \cdot \Pr{f(b_i',b_{-i}) = y \mid A_{-i}, b_{-i}}\\
& = \quad \sum_{y \in \mathcal{X}} v_i(y) \cdot \prod_{j} \Pr{f_j(b_i',b_{-i}) = y_{i,j} \mid A_{-i}, b_{-i}}\\
&\ge \quad \frac{1}{\gamma(\mathcal{V})}\cdot \sum_{y \in \mathcal{X}} v_i(y) \cdot \Pr{f(b_i^{A_i},b_{-i}) = y \mid A_{-i}, b_{-i}}\\
&= \quad \frac{1}{\gamma(\mathcal{V})} \cdot \E\left[ v_i(f(b_i^{A_i},b_{-i})) \mid A_{-i}, b_{-i} \right] \\
& = \quad \E\left[ \bar{v}_i(f(b_i^{A_i},b_{-i})) \mid A_{-i}, b_{-i} \right] \enspace.
\end{align*}
In addition, because payments are simply additive across mechanisms, it is straightforward to see that for every bidder $i$
\[
\E\left[ p_{i}(b'_i,b_{-i}) \mid A_{-i},b_{-i} \right] = \E\left[ p_{i}(b_i^{A_i},b_{-i}) \mid A_{-i},b_{-i} \right].
\]
This allows to apply the smoothness bound for Bayesian mechanisms with independent types from~\cite{SyrgkanisT13} to derive
\begin{align*}
\nonumber &\sum_{i} \E\left[ u_i(b'_i,b_{-i}) \right] \\
& = \quad \sum_{i} \E\left[v_i(f(b'_i,b_{-i}))\right] - \E\left[p_i(b'_i,b_{-i}) \right] \\
& \ge \quad \sum_{i} \E\left[ \bar{v}_i(f(b_i^{A_i},b_{-i}))\right] - \E\left[ p_i(b_i^{A_i},b_{-i}) \right] \\
& \ge \quad \lambda\cdot \sum_i \E\left[\bar{v}_i(x^*)\right] - \mu_1 \sum_i \E\left[ p_i(b)\right] - \mu_2 \sum_i \E\left[ h_i(b_i,f(b))\right]\\
& = \quad \frac{\lambda}{\gamma(\mathcal{V})}\cdot \sum_i \E\left[v_i(x^*)\right] - \mu_1 \sum_i \E\left[ p_i(b)\right] - \mu_2 \sum_i \E\left[ h_i(b_i,f(b))\right]
\end{align*}
This proves the desired smoothness guarantee and implies the theorem. \qed
\end{myproof}
%\input{setSubmodularLower}
% !TEX root = main.tex

\section{Composition with Everybody-or-Nobody Admission}

\label{sec:correlated}

We consider the case in which at each point in time each mechanism is either available to all bidders or to none. We let $A_j = A_{i,j}$ for all $i \in [n]$ and $q_j = \Pr{A_j = 1}$. Note that all $A_j$ are assumed to be independent.

Let the social optimum be denoted by $x^*$. We assume that $x_j^* = \bot_j$ if $A_j = 0$. Otherwise, $x^*$ might have different values, depending on the availabilities of other mechanisms. Let us denote the possible outcomes by $x_j^1, x_j^2, \ldots$ and let $r_j^\ell := \Pr{x^*_j = x_j^\ell \growingmid A_j = 1}$. That is, $r_j^\ell$ is the marginal probability of $x_j^\ell$ conditioned on $j$ being available. Theorem \ref{everybody-or-nobody} formulates our main result in this section.

\begin{theorem}\label{everybody-or-nobody} 
  The price of anarchy for oblivious learning for simultaneous comp\-osition of weakly $(\lambda, \mu_1, \mu_2)$-smooth mechanisms with monotone lattice-submodular valuations and everybody-or-nobody admission is at most $4e/(e-1)\cdot(\mu_2 + \max(1,\mu_1))/\lambda^2$.
\end{theorem}

\begin{proof}
We will prove that, for each bidder $i$ and each mechanism $j$ there are randomized deviation strategies $b_{i, j}'$  that are independent of the availabilities such that the following smoothness guarantee holds against any (potentially non-oblivious) bidding strategy $b$:
\begin{align*}
&\sum_i \E\left[ u_i(b_i', b_{-i})\right] \\
&\ge \quad \left( 1 - \frac{1}{\e} \right) \frac{\lambda^2}{4} \sum_i \E\left[ v_i(x^*) \right] - \mu_1 \sum_i \E\left[ p_i(b) \right] - \mu_2 \sum_i \E\left[ h_i(b_i, f(b))\right]\enspace.
\end{align*}
From this guarantee the claim of the theorem again follows as described in Section~\ref{sec:model}.

To define $b_{i, j}'$, every bidder $i$ draws two vectors $z^i$ and $\tilde{t}^i$ at random as follows. He sets $z_j^i$ to $x_j^\ell$ with probability $r_j^\ell / \alpha$, where $\alpha = 2/\lambda$, and to $\bot_j$ with the remaining probability. Furthermore, he sets $\tilde{t}_j^i$ to $x_j^\ell$ with probability $q_j r_j^\ell$ and to $\bot_j$ with the remaining probability. These draws are performed independent of any availabilities. Observe that for each $i$, we have $\Ex{\sum_{i'} v_{i'}(\tilde{t}^i)} \geq (1 - \frac{1}{\e}) \Ex{\sum_{i'} v_{i'}(x^*)}$ by Lemma~\ref{lemma:gapLatticeIndep}.

Due to the random draws, each bidder $\replaced{i}{i'}$ defines functions $w^{\replaced{i}{i'}}_{\replaced{i'}{i}, j}\colon \Omega_j \to \RR$ for each bidder $\replaced{i'}{i}$ and each mechanism $j$. Function $w^{\replaced{i}{i'}}_{\replaced{i'}{i}, j}$ maps an outcome of mechanism $j$, denoted by $y_j$, to a real number as follows
\[
w^{\replaced{i}{i'}}_{\replaced{i'}{i}, j}(y_j) = v_{\replaced{i'}{i}}(\tilde{t}_1^{\replaced{i}{i'}}, \ldots, \tilde{t}_{j - 1}^{\replaced{i}{i'}}, y_j \wedge z_j^{i'}, \bot_{j + 1}, \ldots, \bot_m) - v_{\replaced{i'}{i}}(\tilde{t}_1^{\replaced{i}{i'}}, \ldots, \tilde{t}_{j - 1}^{\replaced{i}{i'}}, \bot_j, \ldots, \bot_m) \enspace.
\]
Note that these functions do not necessarily reflect the actual value any outcome might have. They are only used to define the deviation strategy: bidder $\replaced{i}{i'}$ pretends all bidders $\replaced{i'}{i}$, including himself, would have valuations $w^{\replaced{i}{i'}}_{\replaced{i'}{i}, j}$ for the outcome of mechanism $j$. This gives him a deviation strategy $b_{\replaced{i}{i'}, j}'$ by setting $b_{\replaced{i}{i'}, j}' = b_{\replaced{i}{i'}, j}^\ast(w^{\replaced{i}{i'}}_{1, j}, \ldots, w^{\replaced{i}{i'}}_{n, j})$ as defined by the smoothness of mechanism $j$.

The proofs for the following three lemmas are presented in Appendix~\ref{app:first_lemma},~\ref{app:second_lemma},~\ref{app:third_lemma}.

\begin{lemma}\label{first_lemma} 
For every bidder $i$ and deviating bids $b_{i, j}' = b_{i, j}^\ast(w^{i}_{1, j}, \ldots, w^{i}_{n, j})$,
\[
\E\left[ v_i(f(b_i', b_{-i}))\right] \quad \geq \quad \sum_j \E\left[ w^i_{i, j}(f_j(b_{i,j}',b_{-i})) \right] - \frac{1}{\alpha (\alpha + 1)} \E\left[ v_i(\tilde{t}^i)\right] \enspace.
\]
\end{lemma}

\begin{lemma}
For the adjusted functions $w$ we can apply smoothness to obtain
\label{second_lemma} 
\begin{multline*}
\sum_i \sum_j \E\left[ w^i_{i, j}(f_j(b_{i, j}', b_{-i})) - p_{i, j}(b_{i, j}', b_{-i})\right] \quad \\
\geq \quad \lambda \sum_i \sum_j q_j \E\left[w^1_{i, j}(z^1_j)\right] - \mu_1 \sum_i \E\left[p_i(b)\right] -\mu_2 \sum_i \E\left[ h_i(b_i,f(b)) \right] \enspace.
\end{multline*}
\end{lemma}

\begin{lemma}
\label{third_lemma} 
For function $w^1$, random vectors $z^1_j$ and $\tilde{t}^1$, and every mechanism $j$
\[
\sum_j q_j \E\left[ w^1_{i,j}(z^1_j)\right] \quad = \quad \frac{1}{\alpha} \E\left[ v_i(\tilde{t}^1)\right] \enspace.
\]
\end{lemma}

The bound from Lemma~\ref{second_lemma} has striking similarities to the smoothness bound \eqref{eq:smoothness}. However, it is expressed in terms of the functions $w^{\replaced{i}{i'}}_{\replaced{i'}{i}, j}$ rather than the actual valuation functions $v_i$. The other two Lemmas show that, in expectation, these functions are close enough to the functions $v_i$ so that this bound actually suffices to prove the main result:
 
\begin{align*}
& \sum_i \E\left[ u_i(b_i', b_{-i})\right] \quad = \quad \sum_i \E\left[ v_i(f(b_i', b_{-i})) - \sum_j p_{i, j}(b_{i, j}', b_{-i})\right] \\
& \geq \sum_i \sum_j \E\left[ w^i_{i, j}(f_j(b_{i, j}', b_{-i})) - p_{i, j}(b_{i, j}', b_{-i})\right] - \frac{1}{\alpha (\alpha + 1)} \sum_i \E\left[ v_i(\tilde{t}^i) \right] \tag{by Lemma~\ref{first_lemma}} \\
%\end{align*}
%\begin{multline}
& \geq \lambda \sum_i \sum_j q_j \E\left[ w^1_{i, j}(z^1_j) \right] - \mu_1 \sum_i \E\left[ p_i(b) \right] %\\
- \mu_2 \sum_i \E\left[ h_i(b_i,f(b)) \right]\\
& \hspace{0.5cm} - \frac{1}{\alpha (\alpha + 1)} \sum_i \E\left[ v_i(\tilde{t}^1) \right] \tag{by Lemma~\ref{second_lemma}} \\%\end{multline}
%\begin{multline}
& = \sum_i \left( \frac{\lambda}{\alpha} - \frac{1}{\alpha (\alpha + 1)} \right) \E\left[ v_i(\tilde{t}^i)\right] %\\
- \mu_1 \sum_i \E\left[ p_i(b) \right] - \mu_2 \sum_i \E\left[ h_i(b_i,f(b)) \right]\enspace. \tag{by Lemma~\ref{third_lemma}}
%\end{multline}
\end{align*}
By setting $\alpha = \frac{2}{\lambda}$
\begin{align*}
\sum_i \E\left[ u_i(b_i', b_{-i}) \right] \; \geq \; \frac{\lambda^2}{4} \sum_i \E\left[ v_i(\tilde{t}^1) \right] - \mu_1 \sum_i \E\left[ p_i(b) \right] - \mu_2 \sum_i \E\left[ h_i(b_i,f(b)) \right] \\
\hspace{0.5cm} \geq \left( 1 - \frac{1}{\e} \right) \frac{\lambda^2}{4} \sum_i \E\left[ v_i(x^*) \right] - \mu_1 \sum_i \E\left[ p_i(b) \right] - \mu_2 \sum_i \E\left[ h_i(b_i,f(b)) \right]\enspace.
\end{align*}
The last step follows from Lemma~\ref{lemma:gapLatticeIndep}.
\end{proof}

Note that technically the mechanism could be randomized itself. Our results extend to this case in a straightforward way.
% !TEX root = main.tex

\section{A Lower Bound for General XOS Functions}

\label{sec:xosLower}

In this section we consider combinatorial auctions with item bidding and first-price auctions. We can apply the previous analysis, since for each bidder the outcomes form a trivial 2-element lattice -- winning an item is the supremum outcome, not winning is the infimum outcome. In the analysis, observe that each bidder determines a random allocation of items according to the probabilities in the optimum. Based on these allocations, bidders determine the valuations $w^{\replaced{i}{i'}}_{\replaced{i'}{i}, j}$, which in turn form the basis for the deviation. The first-price auction with general bidding space is $(1-1/\e,1,0)$-smooth~\cite{SyrgkanisT13}. If valuation functions are submodular, the composition theorems can be applied to yield the following corollary.

\begin{corollary}
	The price of anarchy for oblivious learning for simultaneous composition of single-item first-price auctions with monotone submodular valuations and fully independent availability is at most $1/(1-1/e)^2$; for everybody-or-nobody admission it is at most $4/(1-1/e)^3$.
\end{corollary}

For more general XOS valuations, we prove a lower bound that with oblivious bidding we will not be able to show a guarantee based on the smoothness parameters -- even for a single bidder, so the bound applies without assumptions on correlation among bidders. The proof can be found in Appendix~\ref{app:lowerbound}.

\begin{theorem}
	\label{thm:xosLower}
  In a simultaneous composition of discrete first-price single-item auctions with $m$ items and XOS valuations, the price of anarchy for pure Nash equilibria with oblivious bidding can be as large as $\Omega((\log m)/(\log \log m))$, while each single mechanism is weakly $(1/2,1,0)$-smooth.
\end{theorem}

\section{Conclusion}

In this paper, we have studied an oblivious variant for no-regret learning in repeated games with incomplete information and proved a composition theorem for smooth mechanisms. The bounds show that even if bidders apply learning algorithms independently of their types, they can still obtain outcomes that approximate the optimal social welfare within a small ratio. 

Our primary motivation are changes over time on the supply side. That is, bidders value items the same at all times but are constrained when they can buy them. A different interpretation that leads to the same model is when bidders value items differently from time to time. Here the valuation for a bundle has the special structure that it is given by the value of a fixed submodular function evaluated on the intersection of this bundle with a random set.

There is potential to generalize this approach to other interesting settings. For example, one could consider general independent types, where the complete availability-vector of a single bidder is drawn from a bidder-specific distribution, and for each bidder this is done independently. In Appendix~\ref{app:changingunitdemand}, we give a partial answer and show how our techniques can be extended to the following case. Consider simultaneous single-item auctions with unit-demand valuations, i.e., $v_i(S) = \max_{j \in S} v_{i, j}$. The distribution over valuations is such that for each item the value $v_{i, j}$ is independently drawn from a distribution of small support. Independent availabilities can be captured in this setting by setting $v_{i, j}$ to a fixed value or to $0$ with the respective probabilities.

\bibliographystyle{abbrv}
\bibliography{bibliography}

\clearpage
\appendix
% !TEX root = main.tex

\section{Applications beyond Auctions}
\label{sec:wireless}

Our results have interesting implications beyond mechanisms that incorporate standard auction formats. A very intriguing one is channel allocation in wireless networks. The overall problem is to maximize the utilization of a wireless channel while avoiding interference. To this end, the following game was defined in \cite{AndrewsD09}: Each player $i$ corresponds to a pair of a sender $s_i$ and a receiver $r_i$. The transmission from $s_i$ to $r_i$ is successful if the signal-to-interference-plus-noise ratio (SINR) is high enough. This means that the incoming interference from senders transmitting simultaneously plus ambient noise is by a factor smaller than the intended signal. Formally, transmission $i$ is successful if
\[
\frac{\frac{p}{d(s_i, r_i)^\alpha}}{\sum_{j \in S \setminus \{ i \}} \frac{p}{d(s_j, r_i)^\alpha} + \nu} \quad \geq \quad \beta \enspace.
\]
Here $p > 0$ is the (fixed) power level, $S \subseteq [n]$ is the set of simultaneous transmissions; $\alpha > 0$, $\beta > 0$, and $\nu \geq 0$ are constants.

To derive a game, each player has two strategies $b_i$: either he decides to transmit or not to. The best possible outcome is a successful transmission. An unsuccessful transmission is the worst possible outcome. Due to the energy consumption, it is considered to be even worse than not transmitting at all. This is reflected in the following utility function.
\[
u_i(b) = \begin{cases}
1 & \text{ if $b_i = 1$ and $i$ is successful against $b_{-i}$} \\
-1 & \text{ if $b_i = 1$ and $i$ is not successful against $b_{-i}$} \\
0 & \text{ if $b_i = 0$}
\end{cases}
\]
The robust price of anarchy of this game is constant~\cite{AsgeirssonM11}. In every coarse correlated equilibrium, the expected number of successful transmissions is only a constant smaller than the maximum possible number of simultaneous successful transmissions.

Quite surprisingly, this game corresponds to a smooth mechanism as follows. Each player decides whether to transmit; a player always has valuation $2$ for making a successful transmission. However, whenever making a transmission (successful or not), the bidder has to pay $1$. This is comparable to an all-pay auction, where each bidder has to pay his bid, regardless of whether he wins the respective item.

\begin{theorem}
	\label{thm:wireless}
	The mechanism representing the channel-allocation game is weakly $(1, \mu_1, \mu_2)$-smooth for $\mu_1 = O(1)$ and $\mu_2 = 0$.
\end{theorem}

\begin{proof}
	Let $S \subseteq N$ be a maximum set of players that can transmit simultaneously. Define $b'$ by setting $b_i' = 1$ for $i \in S$ and $b_i' = 0$ for $i \not\in S$. That is, $u_i(b_i', b_{-i}) = 0$ for all $i \not\in S$. Consider some bid vector $b$, let $T$ be the set of players making a transmission attempt. Note that by our definition $\sum_i p_i(b) = \lvert T \rvert$.

Furthermore, $i \in S$ is successful under $(b_i', b_{-i})$ if and only if
\[
\frac{\frac{p}{d(s_i, r_i)^\alpha}}{\sum_{j \in T \setminus \{ i \}} \frac{p}{d(s_j, r_i)^\alpha} + \nu} \quad \geq \quad \beta \enspace,
\]
for which it is sufficient to have
\[
\sum_{j \in T \setminus \{ i \}} \frac{d(s_i, r_i)^\alpha}{d(s_j, r_i)^\alpha} + \frac{d(s_i, r_i)^\alpha}{p} \nu \quad < \quad \frac{1}{\beta} \enspace,
\]
which is equivalent to
\[
\sum_{j \in T \setminus \{ i \}} a_{j, i} < 1 \qquad \text{ where } a_{j, i} = \min\left\{ 1, \frac{1}{\frac{1}{\beta} - \frac{d(s_i, r_i)^\alpha}{p} \nu} \frac{d(s_i, r_i)^\alpha}{d(s_j, r_i)^\alpha} \right\} \enspace.
\]
This implies $u_i(b_i', b_{-i}) \geq 1 - 2 \sum_{j \in T \setminus \{ i \}} a_{j, i}$. Taking the sum over all $i \in S$, we get
\[
\sum_{i \in S} u_i(b_i', b_{-i}) \quad \geq \quad \lvert S \rvert - 2 \sum_{j \in T} \sum_{i \in S \setminus \{ j \}} a_{j, i} \enspace.
\]
Lemma~11 in \cite{AsgeirssonM11} shows that $\sum_{i \in S \setminus \{ j \}} a_{j, i} \leq C$ for some constant $C$ because $S \setminus \{ j \}$ is a feasible set. This gives us
\[
\sum_{i \in N} u_i(b_i', b_{-i}) \quad = \quad \sum_{i \in S} u_i(b_i', b_{-i}) \quad \geq \quad \lvert S \rvert - 2 \sum_{j \in T} C  \quad = \quad \lvert S \rvert - 2 C \sum_{i \in N} p_i(b) \enspace.
\]
\qed
\end{proof}

By applying our composition theorems, we obtain a constant price of anarchy for oblivious learning in this game even when we have multiple channels with fully independent or everybody-or-nobody availability. This simplifies and generalizes an approach based on sleeping expert learning in~\cite{DamsHK13}. Furthermore, our analysis can also be conducted similarly for other interference models with a bounded independence condition, see~\cite{DamsHK13,DamsHK14} for a discussion.

% !TEX root = main.tex

\section{Extension to Changing Unit-Demand Functions}
\label{app:changingunitdemand}

We now consider a case in which valuations change over time rather than the supply. In particular, we consider a unit-demand setting, i.e., there are values $v_{i, j}$ such that $v_i(S) = \max_{j \in S} v_{i, j}$. We assume that each of the $v_{i, j}$ is an independent random variable in which constantly many outcomes have a positive probability. So, for a fixed player $i$, the valuation is defined such that for $k = 1, \ldots, K$ we let $v_{i, j} = v_{i, j}^{(k)}$ with probability $q_{i, j}^{(k)}$, $\sum_{k = 1}^K q_{i, j}^{(k)} = 1$. Without loss of generality, let $v_{i, j}^{(1)} \geq v_{i, j}^{(2)} \geq \ldots \geq v_{i, j}^{(K)}$.

To apply availability-oblivious learning, player $i$ now makes $K$ copies of each item $j$. The $k$th copy of item $j$ has value $v_{i, j}^{(k)}$, and it is available whenever $v_{i, j} \geq v_{i, j}^{(k)}$. Note that, when restricting the consideration to only the most valuable item, we can equivalently assume that availabilities of items are drawn independently with probability $q_{i, j}^{(k)} / \sum_{k'=k}^K q_{i, j}^{(k')}$ for the $k$th copy of item $j$.

By the same argument as in Section~\ref{sec:independent}, we then have
\begin{align*}
\sum_{i} \E\left[ u_i(b'_i,b_{-i}) \right] & = \sum_{i} \E\left[v_i(f(b'_i,b_{-i}))\right] - \E\left[p_i(b'_i,b_{-i}) \right] \\
& \ge \quad \left( 1 - \frac{1}{e} \right) \sum_{i} \E\left[ v_i(f(b_i^{A_i},b_{-i}))\right] - \E\left[ p_i(b_i^{A_i},b_{-i}) \right] \enspace,
\end{align*}
when comparing the availability-oblivious deviation $b_i'$ with the availability-aware ones $b_i^{A_i}$.

Therefore, if each single mechanism is weakly $(\lambda, \mu_1, \mu_2)$-smooth, the price of anarchy for oblivious learning is at most $e/(e-1)\cdot(\mu_2 + \max(1,\mu_1))/\lambda \cdot$.
\section{Missing Proofs}
% !TEX root = main.tex

\subsection{Proof of Lemma~\ref{lemma:gapLatticeIndep}}\label{app:correlationGap} 
% !TEX root = main.tex

% \begin{lemma}[Correlation Gap on a Product Lattice]
% \label{lemma:gapLatticeIndep}
% Let $v$ be a function with diminishing marginal returns on a product lattice. Given vectors $x^1, \ldots, x^k$ and numbers $\alpha_1, \ldots, \alpha_k \in [0,1]$ such that $\sum_{j = 1}^k \alpha_j = 1$, determine another vector $y$ at random by setting component $y_i$ to $x^j_i$ independently with probability $\alpha_j$. Then $\Ex{v(y)} \geq \left( 1 - \frac{1}{\e} \right) \sum_{j=1}^k \alpha_j v(x^j)$.
% \end{lemma}

% \begin{proof}
Without loss of generality, let $v(x^1) \geq v(x^2) \geq \ldots \geq v(x^k)$. For each component $i \in [m]$, let $J_i \in [k]$ be the random variable of the index of the vector from which $y_i$ was taken.

Let $z$ be defined by
\[
z_i = \begin{cases}
\bot_i & \text{ if $J_i = 1$} \\
y_i & \text{ otherwise}
\end{cases}
\]

If $J_i \neq 1$, we have
\[
v(y_1, \ldots, y_i, z_{i+1}, \ldots, z_m) - v(y_1, \ldots, y_{i-1}, z_i, \ldots, z_m)  = 0\enspace.
\]
Otherwise, if $J_i = 1$, we have
\begin{align*}
& v(y_1, \ldots, y_i, z_{i+1}, \ldots, z_m) - v(y_1, \ldots, y_{i-1}, z_i, \ldots, z_m) \\
& \geq v(x^1_1 \vee y_1, \ldots, x^1_{i-1} \vee y_{i-1}, y_i, z_{i+1}, \ldots, z_m) - v(x^1_1 \vee y_1, \ldots, x^1_{i - 1} \vee y_{i-1}, \bot_i, z_{i + 1}, \ldots, z_m) \\
& = v(x^1_1 \vee y_1, \ldots, x^1_{i-1} \vee y_{i-1}, x^1_i, z_{i+1}, \ldots, z_m) - v(x^1_1 \vee y_1, \ldots, x^1_{i - 1} \vee y_{i-1}, \bot_i, z_{i + 1}, \ldots, z_m).
\end{align*}
That is, in combination, we get
\begin{align*}
&\Ex{v(y_1, \ldots, y_i, z_{i+1}, \ldots, z_m) - v(y_1, \ldots, y_{i-1}, z_i, \ldots, z_m)}\\
&\geq \alpha_1 \Ex{v(x^1_1\vee y_1, \ldots, x^1_{i-1} \vee y_{i-1}, x^1_i, z_{i+1}, \ldots, z_m)\right. \\
&\hspace{5cm} \left. - v(x^1_1 \vee y_1, \ldots, x^1_{i - 1} \vee y_{i-1}, \bot_i, z_{i + 1}, \ldots, z_m) \growingmid J_i = 1}\\ 
&= \alpha_1 \Ex{v(x^1_1 \vee y_1, \ldots, x^1_{i-1} \vee y_{i-1}, x^1_i, z_{i+1}, \ldots, z_m) \right.\\
 &\hspace{5cm} \left. - v(x^1_1 \vee y_1, \ldots, x^1_{i - 1} \vee y_{i-1}, \bot_i, z_{i + 1}, \ldots, z_m)} \enspace,
\end{align*}
where the last step uses the independence of the components.

Note that, by diminishing marginal returns, we have 
\begin{align*}
 &v(x^1_1 \vee y_1, \ldots, x^1_{i-1} \vee y_{i-1}, x^1_i, z_{i+1}, \ldots, z_m) - v(x^1_1 \vee y_1, \ldots, x^1_{i - 1} \vee y_{i-1}, \bot_i, z_{i + 1}, \ldots, z_m)\\
&\geq v(x^1_1 \vee y_1, \ldots, x^1_{i-1} \vee y_{i-1}, x^1_i \vee z_i, z_{i+1}, \ldots, z_m)\\
&\hspace{6.5cm}- v(x^1_1 \vee y_1, \ldots, x^1_{i - 1} \vee y_{i-1}, z_i, z_{i + 1}, \ldots, z_m).
\end{align*} 
Applying furthermore $x^1_i \vee z_i = x^1_i \vee y_i$ and taking the expectation, we get
\begin{align*}
& \Ex{v(y_1, \ldots, y_i, z_{i+1}, \ldots, z_m) - v(y_1, \ldots, y_{i-1}, z_i, \ldots, z_m)}\\
& \geq \alpha_1 \mbox{\rm\bf E}\left[v(x^1_1 \vee y_1, \ldots, x^1_{i-1} \vee y_{i-1}, x^1_i \vee y_i, z_{i+1}, \ldots, z_m) \right.\\
&\hspace{6.2cm}\left. - v(x^1_1 \vee y_1, \ldots, x^1_{i - 1} \vee y_{i-1}, z_i, z_{i + 1}, \ldots, z_m)\right].
\end{align*}
Overall, we get
\begin{align*}
&\Ex{v(y)} = \Ex{v(z) + \sum_{i = 1}^m v(y_1, \ldots, y_i, z_{i+1}, \ldots, z_m) - v(y_1, \ldots, y_{i-1}, z_i, \ldots, z_m) } \\
& \geq \Ex{v(z)} + \sum_{i = 1}^m \alpha_1 \Ex{ v(x^1_1 \vee y_1, \ldots, x^1_i \vee y_i, z_{i+1}, \ldots, z_m)\right. \\
&\hspace{5.9cm}\left. - v(x^1_1 \vee y_1, \ldots, x^1_{i - 1} \vee y_{i-1}, z_i, \ldots, z_m)  } \\
& = \Ex{v(z)} + \alpha_1 \left( \Ex{v(x^1 \vee y)} - \Ex{v(z)} \right) \\
& \geq \left( 1 - \alpha_1 \right) \Ex{v(z)} + \alpha_1 v(x^1)\enspace.
\end{align*}
By applying this argument inductively, we also get
\[
\Ex{v(y)} \geq \sum_{i=1}^k \alpha_i v(x^1) \prod_{i' = 1}^{i-1} (1- \alpha_{i'}) \enspace.
\]
As we assumed $v(x^1) \geq v(x^2) \geq \ldots \geq v(x^k)$, the FKG inequality gives us
\[
\sum_{i=1}^k \alpha_i v(x^1) \prod_{i' = 1}^{i-1} (1- \alpha_{i'}) \geq \left( \sum_{i=1}^k \alpha_i \prod_{i' = 1}^{i-1} (1- \alpha_{i'}) \right) \left(\sum_{i=1}^k \alpha_i v(x^1) \right) \enspace.
\]
This implies
\[
\Ex{v(y)} \geq \left( \frac{1}{k} \sum_{i=1}^k \left( 1 - \frac{1}{k} \right)^{i-1} \right) \left(\sum_{i=1}^k \alpha_i v(x^1) \right) \geq \left( 1 - \frac{1}{\e} \right) \sum_{i=1}^k \alpha_i v(x^1) \enspace. \qed
\]

\subsection{Proof of Lemma~\ref{first_lemma}}\label{app:first_lemma}
%\begin{proof}[of Lemma~\ref{first_lemma}]
Let $y^i_j = f_j(b_{i, j}', b_{-i}) \wedge z^i_j$ and $\hat{z}^i_j = z^i_j$ if $j$ is available and $\hat{z}^i_j = \bot_j$ otherwise. Notice that $f_j(b_{i, j}', b_{-i}) \wedge \hat{z}_j^i = f_j(b_{i, j}', b_{-i}) \wedge z_j^i$. This is because $\hat{z}_j^i = z_j^i$ when $j$ is available and $f_j(b_{i, j}', b_{-i}) = \bot_j$ when $j$ is not available. 

By monotonicity, we have $v_i(f(b_i', b_{-i})) \geq v_i(f(b_i', b_{-i}) \wedge z^i) = v_i(y^i)$. Furthermore, we can decompose $v_i(y^i)$ into a telescoping sum by
\[
v_i(y^i) = \sum_j v_i(y^i_1, \ldots, y^i_j, \bot_{j+1}, \ldots, \bot_m) - v_i(y^i_1, \ldots, y^i_{j-1}, \bot_j, \ldots, \bot_m) \enspace.
\]
Next, we bound each of these terms independently using diminishing marginal returns multiple times
\begin{align*}
& v_i(y^i_1, \ldots, y^i_j, \bot_{j+1}, \ldots, \bot_m) - v_i(y^i_1, \ldots, y^i_{j-1}, \bot_j, \ldots, \bot_m) \\
& \geq v_i(\hat{z}^i_1, \ldots, \hat{z}^i_{j-1}, y^i_j, \bot_{j+1}, \ldots, \bot_m) - v_i(\hat{z}^i_1, \ldots, \hat{z}^i_{j-1}, \bot_j, \ldots, \bot_m) \\
& \geq v_i(\hat{z}^i_1 \vee \tilde{t}^i_1, \ldots, \hat{z}^i_{j-1} \vee \tilde{t}^i_{j-1}, y^i_j, \bot_{j+1}, \ldots, \bot_m) - v_i(\hat{z}^i_1 \vee \tilde{t}^i_1, \ldots, \hat{z}^i_{j-1} \vee \tilde{t}^i_{j-1}, \bot_j, \ldots, \bot_m) \\
& = v_i(\hat{z}^i_1 \vee \tilde{t}^i_1, \ldots, \hat{z}^i_{j-1} \vee \tilde{t}^i_{j-1}, \hat{z}^i_j, \bot_{j+1}, \ldots, \bot_m) - v_i(\hat{z}^i_1 \vee \tilde{t}^i_1, \ldots, \hat{z}^i_{j-1} \vee \tilde{t}^i_{j-1}, \bot_j, \ldots, \bot_m) \\
& \quad - \left(v_i(\hat{z}^i_1 \vee \tilde{t}^i_1, \ldots, \hat{z}^i_{j-1} \vee \tilde{t}^i_{j-1}, \hat{z}^i_j, \bot_{j+1}, \ldots, \bot_m) \right.\\
&\hspace{6.3cm} \left. - v_i(\hat{z}^i_1 \vee \tilde{t}^i_1, \ldots, \hat{z}^i_{j-1} \vee \tilde{t}^i_{j-1}, y^i_j, \bot_{j+1}, \ldots, \bot_m) \right) \\
& \geq v_i(\hat{z}^i_1 \vee \tilde{t}^i_1, \ldots, \hat{z}^i_{j-1} \vee \tilde{t}^i_{j-1}, \hat{z}^i_j, \bot_{j+1}, \ldots, \bot_m) - v_i(\hat{z}^i_1 \vee \tilde{t}^i_1, \ldots, \hat{z}^i_{j-1} \vee \tilde{t}^i_{j-1}, \bot_j, \ldots, \bot_m) \\
& \quad - \left(v_i(\tilde{t}^i_1, \ldots, \tilde{t}^i_{j-1}, \hat{z}^i_j, \bot_{j+1}, \ldots, \bot_m) - v_i(\tilde{t}^i_1, \ldots, \tilde{t}^i_{j-1}, y^i_j, \bot_{j+1}, \ldots, \bot_m) \right) \\
& = v_i(\hat{z}^i_1 \vee \tilde{t}^i_1, \ldots, \hat{z}^i_{j-1} \vee \tilde{t}^i_{j-1}, \hat{z}^i_j, \bot_{j+1}, \ldots, \bot_m) - v_i(\hat{z}^i_1 \vee \tilde{t}^i_1, \ldots, \hat{z}^i_{j-1} \vee \tilde{t}^i_{j-1}, \bot_j, \ldots, \bot_m) \\
& \quad - \left(v_i(\tilde{t}^i_1, \ldots, \tilde{t}^i_{j-1}, \hat{z}^i_j, \bot_{j+1}, \ldots, \bot_m) - v_i(\tilde{t}^i_1, \ldots, \tilde{t}^i_{j-1}, \bot_j, \ldots, \bot_m) \right) \\
& \quad + v_i(\tilde{t}^i_1, \ldots, \tilde{t}^i_{j-1}, y^i_j, \bot_{j+1}, \ldots, \bot_m) - v_i(\tilde{t}^i_1, \ldots, \tilde{t}^i_{j-1}, \bot_j, \ldots, \bot_m) \enspace.
\end{align*}
That is, by linearity of expectation, we have
\begin{align*}
& \Ex{v_i(y^i_1, \ldots, y^i_j, \bot_{j+1}, \ldots, \bot_m) - v_i(y^i_1, \ldots, y^i_{j-1}, \bot_j, \ldots, \bot_m)} \quad \geq \\
&\underbrace{\Ex{v_i(\hat{z}^i_1 \vee \tilde{t}^i_1, \ldots, \hat{z}^i_{j-1} \vee \tilde{t}^i_{j-1}, \hat{z}^i_j, \bot_{j+1}, \ldots, \bot_m) - v_i(\hat{z}^i_1 \vee \tilde{t}^i_1, \ldots, \hat{z}^i_{j-1} \vee \tilde{t}^i_{j-1}, \bot_j, \ldots, \bot_m)}}_{\text{part 1}}\\
&- \underbrace{\Ex{v_i(\tilde{t}^i_1, \ldots, \tilde{t}^i_{j-1}, \hat{z}^i_j, \bot_{j+1}, \ldots, \bot_m) - v_i(\tilde{t}^i_1, \ldots, \tilde{t}^i_{j-1}, \bot_j, \ldots, \bot_m)}}_{\text{part 2}}\\
&+ \underbrace{\Ex{v_i(\tilde{t}^i_1, \ldots, \tilde{t}^i_{j-1}, y^i_j, \bot_{j+1}, \ldots, \bot_m) - v_i(\tilde{t}^i_1, \ldots, \tilde{t}^i_{j-1}, \bot_j, \ldots, \bot_m)}}_{\text{part 3}}\enspace.
\end{align*}

To simplify part 2, we use the fact that $\tilde{t}^i_1, \ldots, \tilde{t}^i_{j-1}$ and $\hat{z}^i_j$ are independent. Therefore, we have
\begin{align*}
& \Ex{v_i(\tilde{t}^i_1, \ldots, \tilde{t}^i_{j-1}, \hat{z}^i_j, \bot_{j+1}, \ldots, \bot_m) - v_i(\tilde{t}^i_1, \ldots, \tilde{t}^i_{j-1}, \bot_j, \ldots, \bot_m)} \\
& = \sum_\ell \frac{q_j r_j^\ell}{\alpha} \Ex{v_i(\tilde{t}^i_1, \ldots, \tilde{t}^i_{j-1}, x^\ell_j, \bot_{j+1}, \ldots, \bot_m) - v_i(\tilde{t}^i_1, \ldots, \tilde{t}^i_{j-1}, \bot_j, \ldots, \bot_m)} \\
& = \frac{1}{\alpha} \Ex{v_i(\tilde{t}^i_1, \ldots, \tilde{t}^i_j, \bot_{j+1}, \ldots, \bot_m) - v_i(\tilde{t}^i_1, \ldots, \tilde{t}^i_{j-1}, \bot_j, \ldots, \bot_m)} \enspace.
\end{align*}
For the same reason, we can also bound part 1 by using
\begin{align*}
& \Ex{v_i(\hat{z}^i_1 \vee \tilde{t}^i_1, \ldots, \hat{z}^i_j \vee \tilde{t}^i_j, \bot_{j+1}, \ldots, \bot_m) - v_i(\hat{z}^i_1 \vee \tilde{t}^i_1, \ldots, \hat{z}^i_{j-1} \vee \tilde{t}^i_{j-1}, \bot_j, \ldots, \bot_m)} \\
& = \Ex{v_i(\hat{z}^i_1 \vee \tilde{t}^i_1, \ldots, \hat{z}^i_{j-1} \vee \tilde{t}^i_{j-1}, \hat{z}^i_j \vee \tilde{t}^i_j, \bot_{j+1}, \ldots, \bot_m)\right. \\
&\hspace{6.5cm}\left. - v_i(\hat{z}^i_1 \vee \tilde{t}^i_1, \ldots, \hat{z}^i_{j-1} \vee \tilde{t}^i_{j-1}, \hat{z}^i_j , \bot_{j+1}, \ldots, \bot_m)} \\
& \quad + \Ex{v_i(\hat{z}^i_1 \vee \tilde{t}^i_1, \ldots, \hat{z}^i_{j-1} \vee \tilde{t}^i_{j-1}, \hat{z}^i_j, \bot_{j+1}, \ldots, \bot_m) \right.\\
&\hspace{7.3cm}\left. - v_i(\hat{z}^i_1 \vee \tilde{t}^i_1, \ldots, \hat{z}^i_{j-1} \vee \tilde{t}^i_{j-1}, \bot_j, \ldots, \bot_m)} \\
& \leq \Ex{v_i(\hat{z}^i_1 \vee \tilde{t}^i_1, \ldots, \hat{z}^i_{j-1} \vee \tilde{t}^i_{j-1}, \tilde{t}^i_j, \bot_{j+1}, \ldots, \bot_m) \right.\\
&\hspace{7.3cm}\left.- v_i(\hat{z}^i_1 \vee \tilde{t}^i_1, \ldots, \hat{z}^i_{j-1} \vee \tilde{t}^i_{j-1}, \bot_j, \ldots, \bot_m)} \\
& \quad + \Ex{v_i(\hat{z}^i_1 \vee \tilde{t}^i_1, \ldots, \hat{z}^i_{j-1} \vee \tilde{t}^i_{j-1}, \hat{z}^i_j, \bot_{j+1}, \ldots, \bot_m) \right.\\
&\hspace{7.3cm}\left.- v_i(\hat{z}^i_1 \vee \tilde{t}^i_1, \ldots, \hat{z}^i_{j-1} \vee \tilde{t}^i_{j-1}, \bot_j, \ldots, \bot_m)} \\
& = (\alpha + 1) \Ex{v_i(\hat{z}^i_1 \vee \tilde{t}^i_1, \ldots, \hat{z}^i_{j-1} \vee \tilde{t}^i_{j-1}, \hat{z}^i_j, \bot_{j+1}, \ldots, \bot_m) \right.\\
&\hspace{7.3cm}\left.- v_i(\hat{z}^i_1 \vee \tilde{t}^i_1, \ldots, \hat{z}^i_{j-1} \vee \tilde{t}^i_{j-1}, \bot_j, \ldots, \bot_m)},
\end{align*}
which implies
\begin{align*}
& \Ex{v_i(\hat{z}^i_1 \vee \tilde{t}^i_1, \ldots, \hat{z}^i_{j-1} \vee \tilde{t}^i_{j-1}, \hat{z}^i_j, \bot_{j+1}, \ldots, \bot_m) - v_i(\hat{z}^i_1 \vee \tilde{t}^i_1, \ldots, \hat{z}^i_{j-1} \vee \tilde{t}^i_{j-1}, \bot_j, \ldots, \bot_m)} \\
& \geq \frac{1}{\alpha + 1} \Ex{v_i(\hat{z}^i_1 \vee \tilde{t}^i_1, \ldots, \hat{z}^i_j \vee \tilde{t}^i_j, \bot_{j+1}, \ldots, \bot_m) - v_i(\hat{z}^i_1 \vee \tilde{t}^i_1, \ldots, \hat{z}^i_{j-1} \vee \tilde{t}^i_{j-1}, \bot_j, \ldots, \bot_m)}
\end{align*}

Finally, part 3 is precisely the definition of $\Ex{w^i_{i, j}(y^i_j)}$. Therefore, in combination, we get
\begin{align*}
& \Ex{v_i(y^i_1, \ldots, y^i_j, \bot_{j+1}, \ldots, \bot_m) - v_i(y^i_1, \ldots, y^i_{j-1}, \bot_j, \ldots, \bot_m)} \\
& \geq \frac{1}{\alpha + 1} \Ex{v_i(\hat{z}^i_1 \vee \tilde{t}^i_1, \ldots, \hat{z}^i_j \vee \tilde{t}^i_j, \bot_{j+1}, \ldots, \bot_m) - v_i(\hat{z}^i_1 \vee \tilde{t}^i_1, \ldots, \hat{z}^i_{j-1} \vee \tilde{t}^i_{j-1}, \bot_j, \ldots, \bot_m)} \\
& \quad - \frac{1}{\alpha} \Ex{v_i(\tilde{t}^i_1, \ldots, \tilde{t}^i_j, \bot_{j+1}, \ldots, \bot_m) - v_i(\tilde{t}^i_1, \ldots, \tilde{t}^i_{j-1}, \bot_j, \ldots, \bot_m)} \\
& \quad + \Ex{w^i_{i, j}(y^i_j)}
\end{align*}

Taking the sum over all $j$, we get two telescoping sums, which simplify to $\Ex{v_i(\hat{z}^i \vee \tilde{t}^i)}$ (part 1) and $\Ex{v_i(\tilde{t}^i)}$ (part 2). This gives us
\begin{align*}
&\Ex{v_i(f(b_i',b_{-i}))} \geq \Ex{v_i(y^i)} \geq \frac{1}{\alpha + 1} \Ex{v_i(\hat{z}^i \vee \tilde{t}^i)} - \frac{1}{\alpha} \Ex{v_i(\tilde{t}^i)} + \sum_j \Ex{w^i_{i, j}(y^i_j)} \\
& \geq \sum_j \Ex{w^i_{i, j}(y^i_j)} - \frac{1}{\alpha (\alpha + 1)} \Ex{v_i(\tilde{t}^i)} = \sum_j \Ex{w^i_{i, j}(f(b_i',b_{-i}))} - \frac{1}{\alpha (\alpha + 1)} \Ex{v_i(\tilde{t}^i)}\enspace.
\end{align*}
\qed
%\end{proof}

\subsection{Proof of Lemma~\ref{second_lemma}}\label{app:second_lemma}
%\begin{proof}[of Lemma~\ref{second_lemma}]
Note that functions $w^{\replaced{i}{i'}}_{\replaced{i'}{i}, j}$ are identically distributed for different $\replaced{i}{i'}$ and independent of any availabilities. Therefore, we have
\begin{equation}
\Ex{w^{\replaced{i}{i'}}_{i, j}(f_j(b_{i, j}', b_{-i}) \growingmid A_j = 1} = \Ex{w^1_{i, j}(f_j(b_{i, j}^\ast(w^1_{1, j}, \ldots, w^1_{n, j}), b_{-i})) \growingmid A_j = 1}
\label{eq:fully-correlated:synchronizew}
\end{equation}
and
\begin{equation}
\Ex{p_{i, j}(b_{i, j}', b_{-i}) \growingmid A_j = 1} = \Ex{p_{i, j}(b_{i, j}^\ast(w^1_{1, j}, \ldots, w^1_{n, j}), b_{-i}) \growingmid A_j = 1} \enspace.
\label{eq:fully-correlated:synchronizep}
\end{equation}

Next, we apply the smoothness of each separate mechanism. Let us first assume mechanism $j$ is available and $z^1$ and $\tilde{t}^1$ are fixed arbitrarily. This also fixes the functions $w^1_{1, j}, \ldots, w^1_{n, j}$. We pretend these are the actual valuation functions. Then smoothness gives us 
\begin{align*}
& \sum_i w^1_{i, j}(f_j(b_{i, j}^\ast(w^1_{1, j}, \ldots, w^1_{n, j}), b_{-i})) - p_{i, j}(b_{i, j}^\ast(w^1_{1, j}, \ldots, w^1_{n, j}), b_{-i}) \\
& \geq \lambda \left(\max_{y \in \Omega_j} \sum_i w^1_{i, j}(y) \right) - \mu_1 \sum_i p_{i, j}(b) - \mu_2 \sum_i h_{i,j}(b_i,f(b))\\
& \geq \lambda \left(\sum_i w^1_{i, j}(z^1_j) \right) - \mu_1 \sum_i p_{i, j}(b) - \mu_2 \sum_i h_{i,j}(b_i,f(b)) \enspace.
\end{align*}

Taking the expectation over $z^1$ and $\tilde{t}^1$, we can combine this bound with \eqref{eq:fully-correlated:synchronizew} and \eqref{eq:fully-correlated:synchronizep} to get
\begin{align*}
& \sum_i \Ex{w^i_{i, j}(f_j(b_{i, j}', b_{-i})) - p_{i, j}(b_{i, j}', b_{-i}) \growingmid A_j = 1} \\
& = \sum_i \Ex{w^1_{i, j}(b_{i, j}^\ast(w^1_{1, j}, \ldots, w^1_{n, j}), b_{-i})) - p_{i, j}(b_{i, j}^\ast(w^1_{1, j}, \ldots, w^1_{n, j}), b_{-i}) \growingmid A_j = 1} \\
& \geq \Ex{ \lambda \sum_i w^1_{i, j}(z^1_j) - \mu_1 \sum_i p_{i, j}(b) - \mu_2 \sum_i h_{i,j}(b_i,f(b)) \growingmid A_j = 1} \enspace.
\end{align*}
If $j$ is not available, then naturally $f_j(b_{i, j}', b_{-i}) = \bot_j$ and $p_{i, j}(b_{i, j}', b_{-i}) = p_{i, j}(b) = 0$. By definition, however, $w^1_{i, j}(z^1_j)$ is independent of the fact whether $j$ is available or not. Therefore, we have
\begin{align*}
& \sum_i \Ex{w^i_{i, j}(f_j(b_{i, j}', b_{-i})) - p_{i, j}(b_{i, j}', b_{-i})} \\
& \geq \quad q_j \Ex{ \lambda \sum_i w^1_{i, j}(z^1_j) \growingmid A_j = 1}%\\ &\hspace{5.2cm}
- q_j \Ex{\mu_1 \sum_i p_{i, j}(b) + \mu_2 \sum_i h_{i,j}(b_i,f(b)) \growingmid A_j = 1} \\
& = \quad q_j \Ex{ \lambda \sum_i w^1_{i, j}(z^1_j)} - \Ex{\mu_1 \sum_i p_{i, j}(b)} - \Ex{\mu_2 \sum_i h_{i, j}(b_i, f(b))} \enspace.
\end{align*}
We can take the sum over all $j$ to get
\begin{align}
& \sum_i \sum_j \Ex{w^i_{i, j}(f_j(b_{i, j}', b_{-i})) - p_{i, j}(b_{i, j}', b_{-i})} \nonumber \\
& \geq \lambda \sum_i \sum_j q_j \Ex{w^1_{i, j}(z^1_j)} - \mu_1 \sum_i \Ex{p_i(b)} - \mu_2 \sum_i \Ex{h_i(b_i, f(b))}\enspace. \qed
\label{eq:fully-correlated:smoothnessw}
\end{align}
%\end{proof}

\subsection{Proof of Lemma~\ref{third_lemma}}\label{app:third_lemma}
%\begin{proof}[of Lemma~\ref{third_lemma}]
By only plugging in the definitions of $w^1_{i, j}$, $z^1$, and $\tilde{t}^1$, we get
\begin{align*}
& \sum_j q_j \Ex{w^1_{i, j}(z^1_j)} \\
& = \sum_j q_j \Ex{v_i(\tilde{t}_1^1, \ldots, \tilde{t}_{j - 1}^1, z_j^1, \bot_{j + 1}, \ldots, \bot_m) - v_i(\tilde{t}_1^1, \ldots, \tilde{t}_{j - 1}^1, \bot_j, \ldots, \bot_m)} \\
& = \sum_j q_j \sum_\ell \frac{r_\ell}{\alpha} \Ex{v_i(\tilde{t}_1^1, \ldots, \tilde{t}_{j - 1}^1, x_j^\ell, \bot_{j + 1}, \ldots, \bot_m) - v_i(\tilde{t}_1^1, \ldots, \tilde{t}_{j - 1}^1, \bot_j, \ldots, \bot_m)} \\
& = \sum_j \frac{1}{\alpha} \sum_\ell q_j r_\ell \Ex{v_i(\tilde{t}_1^1, \ldots, \tilde{t}_{j - 1}^1, x_j^\ell, \bot_{j + 1}, \ldots, \bot_m) - v_i(\tilde{t}_1^1, \ldots, \tilde{t}_{j - 1}^1, \bot_j, \ldots, \bot_m)} \\
& = \sum_j \frac{1}{\alpha} \Ex{v_i(\tilde{t}^1_1, \ldots, \tilde{t}^1_j, \bot_{j + 1}, \ldots, \bot_m) - v_i(\tilde{t}^1_1, \ldots, \tilde{t}^1_{j - 1}, \bot_j, \ldots, \bot_m)} \\
& = \frac{1}{\alpha} \Ex{v_i(\tilde{t}^1)} \enspace.
\end{align*}
\qed
%\end{proof}

\subsection{Proof of Theorem~\ref{thm:xosLower}}\label{app:lowerbound}
%\begin{proof}
	Consider the following single-item first-price auction with a discrete bidding space. Each bidder has valuation 0 or 2 for the item, and the set of possible bids is $\{0,1,2\}$. The item is sold to one of the bidders with the maximal bid (arbitrary but fixed deterministic tie-breaking), and only this bidder pays his bid. If all bidders bid 0, the item is not given away. It is easy to see that if each bidder $i$ in this auction deviates to half of his valuation, the auction becomes smooth with $\lambda = 1/2$, $\mu_1 = 1$ and $\mu_2 = 0$. Hence, the auction has a price of anarchy of at most 2.
	
	We compose this auction for a set $[m]$ of $m = k^2$ items, for some integer $k > 0$, and every item is sold simultaneously via the first-price auction above. There is a single bidder, and he has an XOS valuation function $v$ as follows. The items are grouped into $k$ groups $M_1,\ldots,M_k$ of $k$ items each. For a set of items $S \subseteq [m]$ we have $v(S) = \max_{\ell=1,\ldots,k} \sum_{j = 1}^m v^{\ell}_j$ with $v^{\ell}_j = 2$ if $j \in M_{\ell}$ and 0 otherwise, for $\ell = 1,\ldots,k$. Consequently, $v(S) = 2 \max_{\ell=1,\ldots,k} |S \cap M_{\ell}|$. We assume each item $j \in [m]$ is available independently with probability $q_j = 1/k$.

If the bidder can deviate depending on the set of available items, a social optimum $b^*$ is obvious -- he considers the group $\ell^*$ with the maximum number of available items and bids $b^*_j = 1$ for all $j \in M_{\ell^*}$ and 0 otherwise. This way he always obtains a set $S$ of items that maximizes the valuation. Furthermore, this is also the best-response since he obtains the maximum valuation at minimum required payment, and the marginal utility of every obtained item is 1. In contrast, we show that every \emph{oblivious} deterministic best-response bid $b$ allows to recover at most a small fraction of the above described optimum. Thus, even the price of anarchy for pure equilibria cannot be bounded by the smoothness guarantee. 

In the optimum $b^*$, the bidder gets all available items from the group with the maximum number of available items. The number of available items in a group follows a binomial distribution $B(k,1/k)$. This scenario is almost identical to throwing $k$ balls uniformly at random into $k$ bins and recording the maximum number of balls in any bin. Now in each bin $k$ balls appear independently at random with probability $1/k$ each, and an almost identical analysis implies
\[ \Ex{v(S(b^*))} = \Theta\left(\frac{\log k}{\log \log k}\right) \enspace.\]

Now consider an oblivious deterministic best-response $b$. The valuation function $v$ treats all items of a group in a symmetric way and all groups in a symmetric way. Let us denote by $r_{\ell}$ the number of items $j \in M_{\ell}$ with $b_j = 1$. For a fixed vector $r$, expected value and payments are the same no matter on which particular items $j \in M_{\ell}$ a bid $b_j = 1$ is placed. For any two groups $M_{\ell}$ and $M_{\ell'}$, the expected valuation and payments remain the same if we change the bid to have $b_j = 1$ for $r_{\ell}$ items in $M_{\ell'}$ and $r_{\ell'}$ items in $M_{\ell'}$. Moreover, the expected payment depends only on $\sum_{\ell} r_{\ell}$. Now suppose there are two groups $M_{\ell}$ and $M_{\ell'}$ such that $r_{\ell},r_{\ell'} \le k/2$. This bidding strategy is obviously dominated by any bid that bids 1 on $r_{\ell} + r_{\ell'}$ items in $M_{\ell}$ and none in $M_{\ell'}$. In conclusion, w.l.o.g.\ we can assume that $r_1 \ge r_2 \ge r_3 \ge \ldots \ge r_k$ and there is $k'$ such that $r_{\ell} \ge k/2$ for $\ell = 1,\ldots,k'-1$, $r_{k'-1} \ge r_{k'} \ge 0$ and $r_{\ell} = 0$ for $\ell = k'+1,\ldots,k$.

We show that every oblivious best-response $b$ has $\Ex{v(S(b))} = O(1)$. Let $p(b)$ denote the total payments, $X_j$ denote the event that item $j$ is available, and $Y_{\ell} = \sum_{j \in M_{\ell}} X_j$ the number of available items in group $M_{\ell}$, for all $\ell=1,\ldots,k$. Note that 
\begin{eqnarray*}
	\Ex{v(S(b)) - p(b)} &=& \Ex{\max_{\ell=1,\ldots,k'} \left(\sum_{j \in M_{\ell}, b_j = 1} 2X_j\right) - \sum_{jÊ\in [m], b_j = 1} X_j}\\
    &\le& 2\Ex{\max_{\ell=1,\ldots,k'} Y_{\ell}} - \frac{k'-1}{2}\enspace.
\end{eqnarray*}    
Further, for any $d=1,\ldots,k$ we can use Chernoff bounds to see
\[
	\Pr{\max_{\ell=1,\ldots,k'} Y_{\ell} \ge d} = 1 - \left( 1 - \Pr{Y_1 \ge d} \right)^{k'} \le 1 - ( 1 - e^{d-1}/d^d)^{k'} \le \min\{1, k'e^{d-1}/d^d \}\enspace.
\]
Hence,
\[
	\Ex{\max_{\ell=1,\ldots,k'} Y_{\ell}} = \sum_{d=1}^k \Pr{\max_{\ell=1,\ldots,k'} Y_{\ell} \ge d} \le \sum_{d=1}^k \min\{1, (k'/e)\cdot(e/d)^d\} \le \frac{3\log k'}{\log \log k'} + \frac{1}{ek'}\enspace.
\]
Thus, $\Ex{v(S(b)) - p(b)} < (6\log k')/(\log \log k') + 6/(ek') - (k'-1)/2$, which is positive only for $k' \le 34$. Every bid $b$ with $k' \ge 35$ is dominated by $b'$ with $b'_j = 0$ for all $j \in [m]$. Hence, for a best-response $b$ we have $\Ex{v(S(b))} < 17$. This proves the theorem.
\qed
%\end{proof}

\end{document}